\documentclass[copyright,creativecommons]{eptcs}
\providecommand{\event}{SynCoP 2016,\\ 3rd International Workshop on Synthesis of Complex Parameters} 

\usepackage{amsmath}
\usepackage{amsthm} 
\usepackage{amsfonts}
\usepackage{subcaption}
\usepackage{tikz}
\usepackage{xifthen}
\usepackage{thmtools}
\usepackage{stmaryrd}
\usetikzlibrary{arrows,automata,positioning}
\declaretheorem[style=definition]{definition} 
\declaretheorem[style=plain]{lemma}
\declaretheorem[style=plain]{theorem}
\declaretheorem[style=plain]{example}

\newcommand{\R}{\mathbb{R}}
\newcommand{\Rpos}{\R_{\geq 0}}
\newcommand{\Q}{\mathbb{Q}}
\newcommand{\Qpos}{\Q_{\geq 0}}
\newcommand{\Reps}{R^{\varepsilon}}
\newcommand{\mcal}{\mathcal}
\newcommand{\veps}{\varepsilon}
\newcommand{\dfunc}{\Rpos^{S \times S}}
\newcommand{\dfuncparsem}{\mcal{D}}
\newcommand{\dfuncparsyn}{\dfuncparsem_{\expr}}
\newcommand{\distparsyn}{d_{\expr}}
\newcommand\minfix{\,\mathrel{\overset{\makebox[0pt]{\mbox{\normalfont\tiny\sffamily $\min$}}}{=}}\,}
\newcommand{\simeps}[1][]{
\ifthenelse{\equal{#1}{}}{\leq_{\veps}}{\leq_{#1}}%
}
\newcommand{\expr}{\mcal{E}}
\newcommand{\equivexpr}{\equiv_\expr}
\newcommand{\distsynorder}{\sqsubseteq_{\mcal{E}}}
\newcommand{\exporder}{\distsynorder}
\newcommand{\minE}{\texttt{MIN}}
\newcommand{\maxE}{\texttt{MAX}}

\newcommand{\WKS}{\mcal{K}=(S,\atom,\lab,\to)}
\newcommand{\PWKS}{\mcal{K_\param}=(S,\atom,\lab,\to)}
\newcommand{\EKS}{\mcal{K}_{\param}^{v}=(S,\atom,\lab,\to_v)}
\newcommand{\lab}{\mcal{L}}
\newcommand{\atom}{AP}
\newcommand{\param}{\mcal{P}}
\newcommand{\eval}{\mcal{V}}
\newcommand{\ffunc}{\mcal{F}}
\newcommand{\ffuncpar}{\ffunc_{\expr}}

\title{Weighted Branching Simulation Distance for Parametric Weighted Kripke Structures}
\author{Louise Foshammer, Kim Guldstrand Larsen and Anders Mariegaard
\institute{Department of Computer Science, Aalborg University\\
Selam Lagerl\"{o}fs Vej 300, DK-9220 Aalborg, Denmark.}
\email{$\{$foshammer,kgl,am$\}$@cs.aau.dk}
}
\def\titlerunning{Weighted Branching Simulation Distance for Parametric Weighted Kripke Structures}
\def\authorrunning{L. Foshammer, K.G Larsen \& A. Mariegaard}
\begin{document}
\maketitle
\begin{abstract}
This paper concerns branching simulation for weighted Kripke structures with parametric weights. Concretely, we consider a weighted extension of branching simulation where a single transition can be matched by a sequence of transitions while preserving the branching behavior. We relax this notion to allow for a small degree of deviation in the matching of weights, inducing a directed distance on states. The distance between two states can be used directly to relate properties of the states within a sub-fragment of weighted CTL. The problem of relating systems thus changes to minimizing the distance which, in the general parametric case, corresponds to finding suitable parameter valuations such that one system can approximately simulate another. Although the distance considers a potentially infinite set of transition sequences we demonstrate that there exists an upper bound on the length of relevant sequences, thereby establishing the computability of the distance. 
\end{abstract}
\section{Introduction}
In recent years within the area of embedded and distributed systems, a significant effort has been made to develop various formalisms for modeling and specification that address non-functional properties. Examples include extensions of classical Timed Automata \cite{DBLP:conf/icalp/AlurD90} with cost and resource consumption/production in Priced Timed Automata \cite{DBLP:conf/hybrid/BehrmannFHLPRV01} and Energy Automata \cite{DBLP:conf/formats/BouyerFLMS08}. For quantitative analysis of these systems, a generalization of bisimulation equivalence by Milner \cite{DBLP:books/daglib/0067019} and Park \cite{Park1981} as behavioral distances \cite{DBLP:journals/jlp/ThraneFL10,DBLP:journals/tcs/LarsenFT11,DBLP:journals/tse/AlfaroFS09} between system, has been studied. 

In parallel, \emph{parametric} extensions of various formalism have been intensively studied. Instead of requiring exact specification of e.g probabilities, cost or timing constraints, these formalisms allow for the use of \emph{parameters} representing unknown or unspecified values. This can be used to encode multiple configurations of the same system as a system being parametric in the configurable quantities. The problem is then to find ``good'' parameter values such that the instantiated system (configuration) performs as expected. For real-time systems, Parametric Timed Automata \cite{DBLP:conf/stoc/AlurHV93,DBLP:journals/ijfcs/AndreCFE09} and Parametric Stateful Timed CSP \cite{DBLP:journals/rts/Andre00D14} have been developed. Parametric probabilistic models \cite{DBLP:conf/cav/HahnHWZ10, DBLP:conf/nfm/HahnHZ11} have also been developed as well as parametric analysis for weighted Kripke structures \cite{christoffersen_et_al:OASIcs:2015:5611, DBLP:conf/lics/EmersonT99,DBLP:journals/fuin/KnapikP14}. \cite{DBLP:conf/lics/EmersonT99} provides an efficient model-checking algorithm for a parametric extension of real-time CTL on timed Kripke structures. \cite{DBLP:journals/fuin/KnapikP14} extends \cite{DBLP:conf/lics/EmersonT99} to full parameter synthesis by demonstrating that model-checking a finite subset of the entire set of parameter values is sufficient.

In this paper we revisit (parametric) weighted Kripke structures with the purpose of lifting the behavioral distance defined in \cite{WCTL_logic} to the parametric setting, demonstrate its fixed point characterization and prove computability of the distance between any two systems. The distance is a generalization of a weighted extension of branching simulation \cite{branching_bisim}. Consider the following two processes $s,t$ both ending in the inactive process 0:
\[
s \to_5  0 \text{ and } t \to_3 t_1 \to_2 0
\]
If $s,t,t_1$ satisfy the same atomic proposition, $t_1$ may be deemed unobservable and $t$ may simulate $s$ as they both evolve into the process 0 with the same overall weight. \cite{WCTL_logic} captures this situation in generality by extending branching simulation with weights.
Consider a similar scenario, where the process $t$ is now parametrized by the parameter $p$:
\[
s \to_5 0 \text{ and } t \to_p t_1 \to_2 0
\]
If $p \neq 3$ we know that $t$ can no longer simulate $s$. However, it should be intuitive that $p = 6$ is somehow worse than $p = 2$ as the latter is closer to 3. Thus, instead of considering pre-orders and Boolean answers we develop a parametric distance between states such that as the value of $p$ approaches $3$, the distance between $s$ and $t$ decreases towards 0. The distance will also give us a direct relation between the properties satisfied by $s$ and $t$ and a distance of 0 implies that any formula satisfied by $s$ is satisfied by $t$. In this way one can reason about how ``close'' a given implementation is to the specification and compare different configurations that are not necessarily able to fully simulate $s$.

The structure of this paper is as follows: in \autoref{sec:prelim} we introduce preliminaries and recall results from \cite{WCTL_logic}, \autoref{sec:WKS_sim_dist} concerns the fixed point characterization of the distance for weighted systems, \autoref{sec:PWKS_dist} lifts the distance to the parametric setting and finally \autoref{sec:conc_future} concludes the paper and describes future work.
\section{Preliminaries}\label{sec:prelim}
A weighted Kripke Structure (WKS) extends the classical Kripke structure by associating to each transition a non-negative rational transition weight.
\begin{definition}[Weighted Kripke Structure]
A weighted Kripke Structure is a tuple $\WKS$ where $S$ is a finite set of states, $\atom$ is a set of atomic propositions, $\lab: S \to \mcal{P}(\atom)$ is a labelling function, associating to each state a set of atomic propositions and $\to \subseteq S \times \Qpos \times S$ is the finite transition relation.
\end{definition}
A transition from $s$ to $s'$ with weight $w$ will be denoted by $s \to_w s'$ instead of $(s,w,s') \in \to$. 
\begin{example}
\autoref{fig:WKSex} depicts the WKS $\WKS$ where $S = \{s,s_1,s_2,s_3,s_4,t,t_1,t_2\}$, $AP = \{a,b\}$, $\lab(s) = \lab(s_1) = \lab(s_2) = \lab(t) = \lab(t_2) = \{a\}$, $\lab(s_3) = \lab(s_4) = \lab(t_1) = \{b\}$ and\\ $\to = \{(s,1,s_1),(s,2,s_2),(s_1,2,s_2),(s_1,1,s_3),(s_1,3,s_4),(s_2,5,s_4),(t,2,t_1),(t,1,t_2),(t_2,2,t_2),(t_2,1,t_1)\}$.
\begin{figure}[ht]
\centering
\begin{tikzpicture}[>=stealth',shorten <=0pt,auto,semithick, every node/.style={scale=0.6},node distance = 2.5cm]
\node[state] (s1)    [label={left:$\{\texttt{a}\}$}] {$s_1$};
\node[state] (s)     [below left of=s1,label={left:$\{\texttt{a}\}$}] {$s$};
\node[state] (s2)    [below right of=s,label={left:$\{\texttt{a}\}$}] {$s_2$};
\node[state] (s3)    [right of=s1,label={right:$\{\texttt{b}\}$}] {$s_3$};
\node[state] (s4)    [right of=s2,label={right:$\{\texttt{b}\}$}] {$s_4$};

\node[state] (t)    [right = 2cm of s3,label={left:$\{\texttt{a}\}$}] {$t$};
\node[state] (t1)   [below left of=t,label={left:$\{\texttt{b}\}$}] {$t_1$};
\node[state] (t2)   [below right of=t,label={right:$\{\texttt{a}\}$}] {$t_2$};

\path[->] (t) edge node [above left] {$2$}  (t1);
\path[->] (t) edge node [above right] {$1$}  (t2);
\path[->] (t2) edge[loop below] node [below] {$2$}  (t2);
\path[->] (t2) edge node [above] {$1$}  (t1);

\path[->] (s) edge node [above] {$1$}  (s1);
\path[->] (s) edge node [below] {$2$}  (s2);
\path[->] (s1) edge node [above] {$1$}  (s3);
\path[->] (s1) edge node [left] {$2$}  (s2);
\path[->] (s1) edge node [right] {$3$}  (s4);
\path[->] (s2) edge node [below] {$5$}  (s4);
\end{tikzpicture}
\caption{WKS $\mcal{K}$ where $s \not\leq t$ and $t \not\leq s$ but $s \simeps[0.5] t$.}
\label{fig:WKSex}
\end{figure}
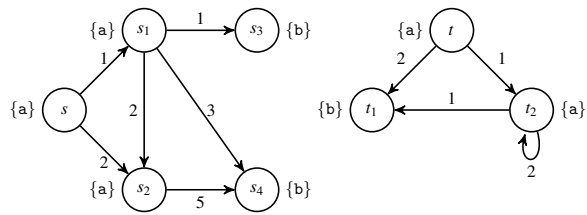
\end{example} 
To reason about behavior of WKSs, we introduce a weighted variant of the classical notion of branching simulation \cite{branching_bisim}. The basic idea is to let a transition $s \to_5 s'$ be matched by a sequence of transitions $t \to_{2} t_1 \to_{2} t_2 \to_{1} t_3$, if $t_3$ can simulate $s'$, as the accumulated weight equals 5. In addition, each intermediate state passed through in the matching transition sequence must be able to simulate $s$. In this way the branching structure of systems is preserved. Instead of always requiring exact weight matching we allow small relative deviations. These small deviations will in \autoref{sec:WKS_sim_dist} induce a directed distance between WKS states.  
\begin{definition}[Weighted Branching $\veps$-Simulation \cite{WCTL_logic}]
Given a WKS $\WKS$ and an $\veps \in \Rpos$, a binary relation $\Reps \subseteq S \times S$ is a weighted branching $\veps$-simulation relation if whenever $(s,t) \in \Reps$:
\begin{itemize}
	\item $\lab(s) = \lab(t)$ 
	\item for all $s \to_w s'$ there exists $t \to_{v_1} t_1 \to_{v_2} \cdots \to_{v_k} t_k$ such that $\sum_{i=1}^k v_i \in [w(1-\veps), w(1+\veps)], (s',t_k) \in \Reps$ and $\forall i<k. (s,t_i) \in \Reps$.
\end{itemize}
\end{definition}
If there exists a weighted branching $\veps$-simulation relating $s$ to $t$ we write $s \simeps t$. If $\veps = 0$ we write $s \leq t$ instead of $s \simeps[0] t$. Note that in this case $\sum_{i=1}^k v_i = w$.

\begin{example}\label{ex:sim}
Consider again \autoref{fig:WKSex} and the pair $(s,t)$. It is clear that $t \not\leq s$ because of the loop $t_2 \to_2 t_2$. We can also observe that $s \not\leq t$ as the transition $s \to_2 s_2$ can only be matched by $t \to_2 t_1$ but $s_2 \not\leq t_1$ as $\lab(s_2) \neq \lab(t_1)$. If we relax the matching requirements by 50\%, we get that $s$ can be simulated by $t$ i.e $s \simeps[0.5] t$; $s \to_2 s_2$ can be matched by $t \to_1 t_2$ as $[2(1-0.5),2(1+0.5)] = [1,3]$ and $1 \in [1,3]$ (another legal match would be $t \to_1 t_2 \to_2 t_2$). Now, $s_2 \to_5 s_4$ can be matched exactly by $t_2 \to_2 t_2 \to_2 t_2 \to_1 t_1$. It follows that $\veps \geq 0.5 \iff s \simeps[\veps] t$.  
\end{example}
If we restrict weighted CTL to only encompass the existential quantifier and remove the next-operator and we know that $s \simeps t$, then for any property $\phi$ of $s$, there exists a related property $\phi^{\veps}$ of $t$.
 \begin{definition}[Existential Fragment of Weighted CTL without next]
The syntax of $EWCTL_{-X}$ is given by the following abstract syntax:
\[
\phi ::= a \mid \neg a \mid \phi_1 \land \phi_2 \mid \phi_1 \lor \phi_2 \mid E(\phi_1U_I\phi_2),
\]
where $a \in \atom$, $I =[l,u]$ and $l,u \in \Qpos$ such that $l \leq u$.
For a WKS $\WKS$ and an arbitrary state $s \in S$, the semantics of $EWCTCL_{-X}$ formulae is given by a satisfiability relation, inductively defined on the structure of formulae in $EWCTL_{-X}$. For existential until; $\mcal{K},s \models E(\phi_1 U_I \phi_2) \iff$ there exists a sequence $s \to_{w_1} s_1 \to_{w_2} \cdots \to_{w_k} s_k \to_{w_{k+1}} \ldots$ where $s_k \models \phi_2, \forall i < k. s_i \models \phi_1$ and $\sum_{i=1}^k w_i \in I$. Let the \emph{$\veps$-expansion} of a formula $\phi = E(\phi_1U_{[l,u]}\phi_2)$ be given by $\phi^{\veps} = E(\phi_1^{\veps}U_{[l(1-\veps),u(1+\veps)]}\phi_2^{\veps})$ where $\phi_1^{\veps}$ and  $\phi_2^{\veps}$ are defined inductively by relaxing any interval by $\veps$ percent in both directions (just as for $[l,u]$).
\end{definition}
\begin{theorem}\label{thm:sim_logic_relation}
\cite{WCTL_logic} Let $\WKS$ be a WKS. Then for all $s,t \in S, \veps \in \Rpos$:
\[
s \simeps t \quad \text{iff} \quad \forall \veps' \in \Qpos, \veps \leq \veps'. [\forall \phi \in EWCTL_{-X}. s \models \phi \implies t \models \phi^{\veps'}].
\]
\end{theorem}
\section{Weighted Branching Simulation Distance for WKSs}\label{sec:WKS_sim_dist}
We now define a directed distance between WKS states as a least fixed point to a set of equations. The distance from $s$ to $t$, $d(s,t)$, represents the minimal $\veps$ such that $s \simeps t$. Thus, if $d(s,t) = 0$ then $s \leq t$. As the distance is based upon weighted branching $\veps$-similarity and its relative deviation in weight matching, it will not satisfy the triangle inequality and is therefore not a hemi-metric.

The distance definition follows intuitively weighted branching $\veps$-simulation. If $s \simeps t$ then no matter what transition $s$ chooses, $t$ has a matching transition sequence with a relative difference of at most $\veps$. In order words, for a given transition $s \to_w s'$, the goal of $t$ is to find a matching sequence $t\rightarrow_{v_1}t_1 \cdots \rightarrow_{v_n} t_n$ that \emph{minimizes} the relative difference $\left|\frac{\sum_{i=1}^n v_i}{w}-1\right|$ as well as ensuring that any intermediate state $t_i$ has as small a distance to $s$ as possible. The strategy of $s$ is then to find a \emph{maximal} move, given the minimization strategy of $t$. In the remainder of this section we assume a fixed WKS $\WKS$.
\begin{definition}[Weighted Branching Simulation Distance]
For an arbitrary pair of states $s,t \in S$ we define the weighted branching simulation distance from $s$ to $t$, $d(s,t)$, as the least fixed point ($\minfix$) of the following set of equations: 
\[
	d(s,t) \minfix \left\{\begin{array}{ll}
														\infty & \text{ if } \lab(s) \neq \lab(t)\\
														\max_{s \rightarrow_w s'} \left\{\min_{t\rightarrow_{v_1}t_1 \cdots \rightarrow_{v_n} t_n} \left\{\max\left\{\begin{array}{l}
														\left|\frac{\sum_{i=1}^n v_i}{w}-1\right|,d(s',t_n),\\
														\max \{d(s,t_i) | \,i < n\}
													 \end{array}
										\right\}
														\right\}\right\} &\text{ o.w}
													  \end{array}\right.
\]
\end{definition}
We assume the empty transition sequence to have accumulated weight 0 and let $\Rpos = \{w\,|\, w \in \mathbb{R}, w \geq 0\} \cup \{\infty\}$ denote the extended set of non-negative reals.
For any $d_1,d_2 \in \dfunc$ let $d_1 \leq d_2$ iff $\forall (s,t) \in S \times S. d_1(s,t) \leq d_2(s,t)$. Then $(\dfunc, \leq)$ constitutes a complete lattice. We now define a monotone function on $(\dfunc,\leq)$ that iteratively refines the distance:
\begin{definition}\label{def:f}
Let $\ffunc : \dfunc \to \dfunc$ be defined for any $d \in \dfunc$:
\[
	\ffunc(d)(s,t) = \left\{\begin{array}{ll}
														\infty & \text{ if } \lab(s) \neq \lab(t)\\
														\max_{s \rightarrow_w s'} \left\{\min_{t\rightarrow_{v_1}t_1 \cdots \rightarrow_{v_n} t_n} \left\{\max\left\{\begin{array}{l}
														\left|\frac{\sum_{i=1}^n v_i}{w}-1\right|,d(s',t_n),\\
														\max \{d(s,t_i) | \,i < n\}
													 \end{array}
										\right\}
														\right\}\right\} & \text{ o.w}
													  \end{array}\right.
\]
\end{definition}
By Tarski's fixed point theorem \cite{tarski} we are guaranteed the existence of a least (pre-)fixed point. Thus, the weighted branching simulation distance is well-defined. Note that any transition $s \to_w s'$, $t$ may have an infinite set of possible transition sequence matches in the presence of cycles in the system. To this end we demonstrate an upper bound, $N$, on the length of relevant matching sequences. As the set of sequences of length at most $N$ is finite (the WKS is finite) computability of the distance follows. The first step is proving that any sequence exercising a loop with accumulated weight 0 can be ignored. We refer to these cycles as \emph{0-cycles}.

\begin{lemma}\label{lem:zerocycle}
For a given move $s \to_w s'$, any transition sequence $t\rightarrow_{v_1}t_1 \cdots \rightarrow_{v_n} t_n$ with a 0-cycle can be removed without affecting the distance $d(s,t)$.
\end{lemma}
\begin{proof}
A transition sequence with one or more 0-cycles has the exact same accumulating weight as the corresponding sequence with no 0-cycles. Furthermore, exercising the loop (once) can only introduce new states, leading to a potentially larger value of $\max \{d(s,t_i) | \,i < n\}$. 
Thus, 0-cycles can be ignored.
\end{proof}
Given that 0-cycles can be removed, we now prove an upper bound $N$ on the length of sequences that affect the distance $d(s,t)$. Thus, any sequence longer than $N$ can be safely ignored.
\begin{lemma}\label{lem:finsequence}
Given that $\mcal{K}$ has no 0-cycles, it is the case that whenever $s \to_w s'$:
\begin{align*}
\exists N. &\forall \pi = t\rightarrow_{v_1}t_1 \ldots \rightarrow_{v_n} t_n, n \geq N.\\
&\exists \pi^* = t\rightarrow_{u_1} t_1' \ldots \rightarrow_{u_m} t_m', m \leq N.\\
&t_n = t_m' \,\land\, \left|\frac{\sum_{i=1}^m u_i}{w}-1\right| \leq \left|\frac{\sum_{i=1}^n v_i}{w}-1\right| \land\\
& \{t_1', \ldots, t_{m-1}'\} \subseteq \{t_1,\ldots,t_{n-1}\}
\end{align*}
\end{lemma}
\begin{proof}
Let $w_{\min} = \min\{w\,|\,s \to_w s'\}$ be the minimum weight in the WKS and let
$s_{w_{\max}} = \max\{w\,|\,s \to_w s'\}$ be the maximum weight out of $s$. We now demonstrate that $N \geq \frac{2 \cdot {s_{w_{\max}}}}{w_{\min}} \cdot |S|$ is sufficient. Any sequence of length $|S|$ must have a loop which, by assumption, cannot have accumulated weight 0. Thus, after $|S|$ transitions, the accumulated weight must be at least $w_{\min}$. Without loss of generality, assume that it is \emph{exactly} $w_{\min}$. If the sequence exercises the loop a number of time, the accumulated weight will at some point reach $2 \cdot s_{w_{\max}}$. Let this sequence be $\pi = t\rightarrow_{v_1}t_1 \cdots \rightarrow_{v_k} t_k$ and let $x$ denote the number of times the loop is exercised i.e $x \cdot w_{\min} \geq 2 \cdot s_{w_{\max}}$. Consider now the corresponding sequence $\pi^* = t\rightarrow_{u_1} t_1' \cdots \rightarrow_{u_l} t_l'$ where the loop is removed. As $\sum_{i=1}^k v_i \geq 2 \cdot s_{w_{\max}}$ it follows that $\left|\frac{\sum_{i=1}^k v_i}{s_{w_{\max}}}-1\right| > 1$. By assumption, removing the loop results in a strictly lower accumulated weight implying $\left|\frac{\sum_{i=1}^l u_i}{s_{w_{\max}}}-1\right| < \left|\frac{\sum_{i=1}^k v_i}{s_{w_{\max}}}-1\right|$. We also directly have $t_k = t_l'$ and $\{t_1,\ldots,t_l'\} \subseteq \{t_1,\ldots,t_k\}$. We will now derive $N$ from the inequality $x \cdot w_{\min} \geq 2 \cdot s_{w_{\max}}$. The number of times the loops is exercised must be equal to the length of the entire sequence divided by $|S|$ as we are sure to exercise the loop every $|S|$ states. Thus, $x = \frac{N}{|S|} \implies \frac{N}{|S|} \cdot w_{\min} \geq 2 \cdot s_{w_{\max}}$ and finally,
\[
N \geq \frac{2 \cdot s_{w_{\max}}}{w_{\min}} \cdot |S|.
\]
\end{proof}

\begin{theorem}[Computability]\label{thm:wdistcomputable}
For two states $s,t \in S$, the weighted branching simulation distance is computable.
\end{theorem}
\begin{proof}
\autoref{lem:finsequence} provides an upper bound on the length of transition sequence that we need to consider in the computation of $d(s,t)$ for any states $s,t \in S$ under the assumption that there are no 0-cycles. By \autoref{lem:zerocycle} we know that any 0-cycles can be removed without affecting the distance. Thus when computing the distance we know for the sub-expression
\[
\min_{t\rightarrow_{v_1}t_1 \cdots \rightarrow_{v_n} t_n} \left\{\max\left\{\begin{array}{l}
														\left|\frac{\sum_{i=1}^n v_i}{w}-1\right|,d(s',t_n),\\
														\max \{d(s,t_i) | \,i < n\}
													 \end{array}
										\right\}
														\right\}
\]
that $n \leq \frac{2 \cdot s_{w_{\max}}}{w_{\min}} \cdot |S|$. As the WKS has a finite number of states and a finite transition relation, only a finite number of sequences of finite length exist. Thus we can modify the distance function to only consider these without affecting the computed distance. Thus, the distance must at some point converge as only a finite number of relative distances on the form $\left|\frac{\sum_{i=1}^n v_i}{w}-1\right|$ exists.
\end{proof}
We leave the exact complexity of computing $d(s,t)$ open but note that deciding $d(s,t) = 0$ is NP-complete \cite{WCTL_logic}.
\begin{example}
Consider again \autoref{fig:WKSex} and the computation of $d(s,t)$.
For the transition $s \to_1 s_1$ only one sequence is considered instead of the entire infinite set arising from the loop; $t \to_1 t_2$. As $\left|\frac{3}{1}-1\right| > \left|\frac{1}{1}-1\right|$, even the sequence that only exercises the loop once is worse than just transitioning to $t_2$ directly. This happens because the accumulated matching weight exceeds the weight being matched and the same states are involved in both sequences. Therefore any sequence involving the loop can be ignored. Note that we in this example consider fewer sequences than implied by the upper bound given in \autoref{lem:finsequence}. For $s \to_1 s_1$ the bound would be $\frac{2 \cdot 2}{2} \cdot 8 = 16$ but it should be clear that the loop can be safely ignored.
For the transition $s \to_2 s_2$, there are two relevant matching sequences; $t \to_1 t_2$ and $t \to_1 t_2 \to_2 t_2$. Thus,
\[
d(s,t) \minfix \max\left\{
													\begin{array}{l}
														\max\left\{\left|\frac{1}{1}-1\right|,d(s_1,t_2)\right\},\\[0.1cm]
														\min\left\{\begin{array}{l}
															\max\left\{\left|\frac{1}{2}-1\right|,d(s_2,t_2)\right\},\\[0.1cm]
															\max\left\{\left|\frac{3}{2}-1\right|,d(s_2,t_2),d(s,t_2)\right\}
														\end{array}\right\}
													\end{array}\right\}
\]
It is easily shown that $d(s_2,t_2) = 0$ as $s_2 \to_5 s_4$ can be matched exactly by $t_2 \to_2 t_2 \to_2 t_2 \to_1 t_1$. Thus,
\[
d(s,t) \minfix \max\left\{\frac{1}{2}, d(s_1,t_2),d(s,t_2)\right\}
\]
where
\begin{align*}
&d(s_1, t_2) \minfix \max\left\{\!\!\!\!\!\begin{array}{ll}
																\max\left\{\left|\frac{2}{2}-1\right|,d(s_2,t_2)\right\},\\[0.1cm]
																\max\left\{\left|\frac{1}{1}-1\right|,d(s_3,t_1)\right\},\\[0.1cm]
																\max\left\{\left|\frac{3}{3}-1\right|,d(s_1,t_2),d(s_4,t_1)\right\}
															 \end{array}\!\!\!\!\!\!\right\} \,\text{ and}\! &d(s,t_2) \minfix \max\left\{\!\!\!\!\begin{array}{ll}
																																						\max\left\{\left|\frac{2}{1}-1\right|, d(s_2,t_2)\right\},\\[0.1cm]
																																						\max\left\{\left|\frac{2}{2}-1\right|, d(s_2,t_2)\right\}
																																						\end{array}\!\!\!\!\!\right\}.
\end{align*}
As $s_4 \not\to$, $s_3 \not\to$ and $t_1 \not\to$ it follows that $d(s_4,t_1) = d(s_3,t_1) = 0$, hence
\[
d(s_1,t_2) \minfix \max\left\{\frac{1}{2},d(s_1,t_2)\right\}.
\]
The least solution to this equation is $\frac{1}{2}$ hence $d(s_1,t_2) = d(s,t) = \frac{1}{2}$. From \autoref{ex:sim} we know that $s \simeps t$ for any $\veps \geq 0.5$ i.e for any $\veps \geq d(s,t)$.
\end{example}
Now that we have established the computability of the distance we prove its relation to weighted branching $\veps$-simulation.
\begin{theorem}\label{lem:sim_dist_relation}
For two states $s,t \in S$ and $\veps \in \Rpos$:
\[
d(s,t) \leq \veps \text{ iff } s \simeps t
\]
\end{theorem}
\begin{proof}
$(\implies)$ For this direction we prove that $R^{\veps} = \{(s,t) \,|\, s,t \in S, d(s,t) \leq \veps\}$ is a weighted branching $\veps$-simulation relation. Suppose $(s,t) \in R^{\veps}$. Then $d(s,t) \leq \veps$ and by the fixed point property of $d$,
\[
d(s,t) = \max_{s \rightarrow_w s'} \left\{\min_{t\rightarrow_{v_1}t_0 \cdots \rightarrow_{v_n} t_n} \left\{\max\left\{\begin{array}{l}
														\left|\frac{\sum_{i=1}^n v_i}{w}-1\right|,\\
														\max \{d(s',t_n)\} \cup \{d(s,t_i) | i < n\}
													 \end{array}
										\right\}
														\right\}\right\}
\]
We immediately have that for any transition $s \to_w s'$ there exists a matching transitions sequence $t\rightarrow_{v_1}t_0 \cdots \rightarrow_{v_n} t_n$ such that $\left|\frac{\sum_{i=1}^n v_i}{w}-1\right| \leq \veps$, $d(s',t_n) \leq \veps$ and $\forall i < n. d(s,t_i) \leq \veps$. Thus, by definition of $R^{\veps}$, for any transition $s \to_w s'$ there exists a sufficient matching sequence from $t$ such that $(s',t_n) \in R^{\veps}$ and $(s,t_i) \in R^{\veps}$ for any $i < n$. 

$(\impliedby)$ Let
\[
d^*(s,t) = \left\{\begin{array}{ll}
						\veps &\text{ if } s \simeps t\\
						\infty &\text{ otherwise}
						\end{array}\right.
\]
We now prove that $d$ is a pre-fixed point of $\ffunc$ i.e $\ffunc(d^*)(s,t) \leq d^*(s,t)$ for any pair $(s,t) \in S$. If $s \not\simeps t$ then $d^*(s,t) = \infty$ and there is nothing to prove. If $s \simeps t$ then for any transition $s \to_w s'$ there exists a matching sequence $t\rightarrow_{v_1}t_0 \cdots \rightarrow_{v_n} t_n$ such that $\sum_{i=1}^n v_i \in [w(1-\veps), w(1+\veps)]$, $s' \simeps t_n$ and $s \simeps t_i$ for any $i < n$. We can now argue that
\[
\max_{s \rightarrow_w s'} \left\{\min_{t\rightarrow_{v_1}t_0 \cdots \rightarrow_{v_n} t_n} \left\{\max\left\{\begin{array}{l}
														\left|\frac{\sum_{i=1}^n v_i}{w}-1\right|,\\
														\max \{d^*(s',t_n)\} \cup \{d^*(s,t_i) | i < n\}
													 \end{array}
										\right\}
														\right\}\right\} \leq \veps
\] 
as $\sum_{i=1}^n v_i \in [w(1-\veps), w(1+\veps)]$ is equivalent to $\left|\frac{\sum_{i=1}^n v_i}{w}-1\right| \leq \veps$, $s' \simeps t_n$ implies $d^*(s',t_n) = \veps$ and similarly $d^*(s,t_i) = \veps$ for any $i < n$. As $d^*$ is a pre-fixed point of $\ffunc$ and $d^*(s,t) = \veps$ it must be the case that $d(s,t) \leq \veps$ as $d$ is the \emph{smallest} pre-fixed point of $\ffunc$.
\end{proof}
Combining \autoref{thm:sim_logic_relation} and \autoref{lem:sim_dist_relation} we immediate get a relation between the distance from one state $s$ to another state $t$ and their $EWCTL_{-X}$ properties: 
\[
d(s,t) \leq \veps \quad \text{iff} \quad \forall \veps' \in \Qpos, \veps \leq \veps'. [\forall \phi \in EWCTL_{-X}. s \models \phi \implies t \models \phi^{\veps'}.
\]
\section{Weighted Branching Simulation Distances for Parametric WKSs}\label{sec:PWKS_dist}
We now extend WKS with parametric weights. The lifted parametric distance will be from a WKS to a parametric system and is represented as a parametric expression that can be evaluated to a rational by a \emph{parameter valuation}. If one abstracts multiple configurations of the same system as one parametric system and calculate the parametric distance, evaluating the distance with respect to a parameter valuation then corresponds to calculating the exact distance from a specific configuration (given by the valuation) to the WKS. Thus, instead of working with multiple WKS configurations, one can use a parametric system and compute the parametric distance once.

A parametric weighted Kripke structure (PWKS) extends WKS by allowing transitions to have parametric weights. Let $\param = \{p_1,\ldots,p_n\}$ be a fixed finite set of parameters. A \emph{parameter valuation} is a function mapping each parameter to a non-negative rational; $v : \param \to \Qpos$. The set of all such valuation will be denote by $\eval$.

\begin{definition}[Parametric Weighted Kripke Structure]
A \emph{parametric weighted Kripke structure} is a tuple $\PWKS$, where $S$ is a finite set of states, $\atom$ is a set of atomic propositions, $\lab: S\to \mcal{P}(\atom)$ is a mapping from states to sets of atomic propositions and $\to \subseteq S \times \param \cup \Qpos \times S$ the finite transition relation. 
\end{definition}
Unless otherwise specified, we assume a fixed PWKS $\PWKS$ in the remainder of this section. One can instantiate a PWKS to a WKS by applying a parameter valuation. A PWKS thus represents an infinite set of WKSs.

\begin{definition}
Given a parameter valuation $v \in \eval$, we define the \emph{instantiated WKS} of $\mcal{K}_\param$ under $v$ to be $\EKS$ where
\[
\to_v = \{(s,v(p),s')\mid (s,p,s')\in \to, p \in \param\} \cup \{(s,w,s')\mid (s,w,s')\in \to, w \in \Qpos\}
\]
\end{definition}
For a state $s$ in $\mcal{K}_\param$ let $s[v]$ be the corresponding state in the WKS $\mcal{K}_\param^v$ and let $\simeps$ be lifted to disjoint unions of WKSs in the natural way. 

Given a WKS state $s$, a PWKS state $t$ and $\veps \geq 0$ we now state three interesting problems:
 \begin{enumerate}
	 \item Does there exist a $v \in \eval$ such that $s \simeps t[v]$?
	 \item Can we characterize the set of ``good'' parameter valuation $V = \{v \,|\, v \in \eval, s \simeps t[v]\}$?
	 \item Can we synthesize a valuation $v \in \eval$ that minimizes $\veps$ for $s \simeps t[v]$?
 \end{enumerate}

We will show how to solve (2) by fixed point computations. The result will be a set of linear inequalities over parameters and $\veps$ which has as solution a set of parameter valuations. Instead of considering a concrete $\veps \in \Rpos$, one can let $\veps$ be an extra parameter. Thus, (1) and (3) can be solved by first solving (2) and applying e.g $Z3$ \cite{DBLP:conf/tacas/MouraB08} and $\nu Z$ \cite{DBLP:conf/tacas/BjornerPF15} or similar tools to solve the inequalities and search for solutions that minimize $\veps$.
\begin{example}
Consider \autoref{fig:PWKSex}. From \autoref{ex:sim} we know that $s \leq_{0.5} t[v]$ if $v(p) = 1$. Both $v(p) = 0$ and $v(p) = 2$ imply $s \leq_{1} t[v]$. It turns out that $v(p) = 1$ is the valuation that minimizes $\veps$ for $s \simeps t[v]$.
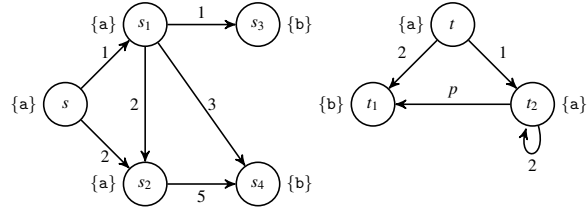
\begin{figure}[ht]
\centering
\begin{tikzpicture}[>=stealth',shorten <=0pt,auto,semithick, every node/.style={scale=0.6},node distance = 2.5cm]
\node[state] (s1)    [label={left:$\{\texttt{a}\}$}] {$s_1$};
\node[state] (s)     [below left of=s1,label={left:$\{\texttt{a}\}$}] {$s$};
\node[state] (s2)    [below right of=s,label={left:$\{\texttt{a}\}$}] {$s_2$};
\node[state] (s3)    [right of=s1,label={right:$\{\texttt{b}\}$}] {$s_3$};
\node[state] (s4)    [right of=s2,label={right:$\{\texttt{b}\}$}] {$s_4$};

\node[state] (t)    [right = 2cm of s3,label={left:$\{\texttt{a}\}$}] {$t$};
\node[state] (t1)   [below left of=t,label={left:$\{\texttt{b}\}$}] {$t_1$};
\node[state] (t2)   [below right of=t,label={right:$\{\texttt{a}\}$}] {$t_2$};

\path[->] (t) edge node [above left] {$2$}  (t1);
\path[->] (t) edge node [above right] {$1$}  (t2);
\path[->] (t2) edge[loop below] node [below] {$2$}  (t2);
\path[->] (t2) edge node [above] {$p$}  (t1);

\path[->] (s) edge node [above] {$1$}  (s1);
\path[->] (s) edge node [below] {$2$}  (s2);
\path[->] (s1) edge node [above] {$1$}  (s3);
\path[->] (s1) edge node [left] {$2$}  (s2);
\path[->] (s1) edge node [right] {$3$}  (s4);
\path[->] (s2) edge node [below] {$5$}  (s4);
\end{tikzpicture}
\caption{A WKS (left) and a PWKS (right)}
\label{fig:PWKSex}
\end{figure}
\end{example}
When lifting the distance to the parametric setting, we consider disjoint unions of systems and require that only the simulating system can be parametric. Let $\mcal{K_\param}=(S_\param,\atom,\lab_\param,\to^\param)$ be a PWKS and $\WKS$ a WKS.
If we were to validate a given parameter valuation we could simply apply the valuation to the PWKS and use $\ffunc$ directly to decide if the distance is below some $\veps$. As we want a full characterization of the good parameter valuation we will instead represent the distance as a function from a pair of states to a function that returns a weighted distance when a parameter valuation is applied; $d: S \times S_\param \to (\eval \to \Rpos)$. We let the set of such function be denoted by $\dfuncparsem$ and define an ordering as follows; for any $d^1,d^2 \in \dfuncparsem$ let $d^1 \leq d^2$ iff $\forall s \in S, t \in S_\param, v \in \eval: d^1(s,t)(v) \leq d^2(s,t)(v)$. Let $\equiv$ denote the set of pairs of semantic equivalent states. Then $(\dfuncparsem,\leq)$ constitutes a complete lattice and we can define a monotone function on $(\dfuncparsem, \leq)$ that iteratively refines the distance:
\begin{definition}\label{def:fpar}
Let $\ffunc : \dfuncparsem \to \dfuncparsem$ be defined for any $d \in \dfuncparsem$:
\[
	\ffunc(d)(s,t) = \left\{\begin{array}{ll}
														\infty & \text{ if } \lab(s) \neq \lab(t)\\
														\max_{s \rightarrow_w s'} \left\{\min_{t\rightarrow_{v_1}t_1 \cdots \rightarrow_{v_n} t_n} \left\{\max\left\{\begin{array}{l}
														\left|\frac{\sum_{i=1}^n v_i}{w}-1\right|,d(s',t_n),\\
														\max \{d(s,t_i) | \,i < n\}
													 \end{array}
										\right\}
														\right\}\right\} & \text{ o.w}
													  \end{array}\right.
\]
\end{definition}
Again, by Tarski's fixed point theorem \cite{tarski} we are guaranteed a least fixed point, denoted by $d_{\min}$. The problem is now that the ordering $\leq$ implies a universal quantification over the entire infinite set of parameter valuations; thus, checking if a fixed point is reached is highly impractical. Instead of representing the distance as a function in valuations we will define it as a \emph{parametric expression} that captures the distance function syntactically. For any two states $s,t$ we associate an syntactic expression $E_{s,t}$ such that the solution set to the inequality $E_{s,t} \leq \veps$ characterizes the set of good parameter valuations i.e applying a parameter valuation to $E_{s,t}$ yields a concrete weighted distance. The syntactic elements for the expressions can be derived directly from $\ffunc$; we need syntax for describing minimums of maximums of basic elements $\left|\frac{v}{w}-1\right|$ and $\infty$ where $w$ is rational and $v$ a linear expression in the parameters. Hence, we define the following abstract syntax:
\[
E_1,E_2 ::= \infty \mid \left|\frac{v}{w}-1\right| \mid \minE\{E_1,E_2\} \mid \maxE\{E_1,E_2\} 
\]
where $w \in \Qpos$ and $v$ is on the form $\sum^{n}_{i=0} a_i p_i + b$ s.t $a_i \in \mathbb{N}$ for all $i < n$ and $b \in \Qpos$.
We extend parameter valuations to expressions in the obvious way and denote by $\llbracket E \rrbracket (v)$ the value of $E$ under $v \in \eval$.
Similar to disjunctive normal form for logical formulae, we assume all expression to be a $\minE$ of $\maxE$'s of basic elements $\left|\frac{v}{w}-1\right|$or $\infty$. To convert an expression, note that for any $v \in \eval$
\[
\llbracket \maxE\{\minE\{E_1,E_2\},E_3\} \rrbracket (v) = \llbracket \minE\{\maxE\{E_1,E_3\},\maxE\{E_2,E_3\}\} \rrbracket (v)
\]
The set of expression on this normal form will be denoted by $\expr$. Now the distance functions can be defined as functions associating to a pair of states a parametric expression; $d_\expr: S \times S_\param \to \expr$. The set of syntactic distance function will be denoted by $\dfuncparsyn$ and the syntactic iterator capturing $d_{\min}$ is defined as follows:   
\begin{definition}\label{def:fparsyn}
Let $\ffuncpar : \dfuncparsyn \to \dfuncparsyn$ be defined for any $d_\expr \in \dfuncparsyn$:
\[
	\ffuncpar(d_\expr)(s,t) = \left\{\begin{array}{ll}
														\infty & \text{ if } \lab(s) \neq \lab(t)\\
														\maxE_{s \rightarrow_w s'} \left\{\minE_{t\rightarrow_{v_1}t_1 \cdots \rightarrow_{v_n} t_n} \left\{\maxE\left\{\begin{array}{l}
														\left|\frac{\sum_{i=1}^n v_i}{w}-1\right|,d_\expr(s',t_n),\\
														\maxE \{d_\expr(s,t_i) | \,i < n\}
													 \end{array}
										\right\}
														\right\}\right\} & \text{ o.w}
													  \end{array}\right.
\]
\end{definition}
We will now define an ordering on elements from $\dfuncparsyn$, by first ordering elements from $\expr$.
\begin{definition}\label{def:exporder}
The syntactic ordering $\exporder \subseteq \expr \times \expr$ is defined inductively on the structure of $\expr$:
\begin{alignat*}{2}
\left|\frac{\sum_{i=1}^n a_ip_i + b}{w}-1\right| \exporder \infty \quad\quad &\text{always}\\
\left|\frac{\sum_{i=1}^n a_ip_i + b}{w}-1\right| \exporder \left|\frac{\sum_{i=1}^n a_i'p_i + b'}{w}-1\right| \quad & \text{iff} \quad && \left\{\begin{array}{lc}
																																																																									\forall i.a_i \leq a_i' \land b \leq b'  &\text{ if } \frac{b'}{w},\frac{b}{w} \geq 1\\
																																																																									\forall i.a_i=a_i' \land b=b' \quad &\text{otherwise}
																																																																									\end{array}\right.\\
\maxE\{E_{1.1},\ldots,E_{1.n}\} \exporder \maxE\{E_{2.1},\ldots,E_{2.m}\} \quad & \text{iff} \quad && \forall i. \exists j. E_{1.i} \exporder E_{2.j}\\
\minE\{E_{1.1},\ldots,E_{1.n}\} \exporder \minE\{E_{2.1},\ldots,E_{2.m}\} \quad & \text{iff} \quad && \forall j. \exists i. E_{1.i} \exporder E_{2.j}
\end{alignat*}
\end{definition}
Let $\equiv_\expr$ be the set of pairs of syntactically equivalent expressions. We now extend the ordering to distance functions:
\begin{definition}
The \emph{syntactic} ordering on distance functions $\distsynorder$ is defined for any $\distparsyn^1, \distparsyn^2 \in \dfuncparsyn$:
\[
\distparsyn^1 \distsynorder \distparsyn^2 \quad \text{ iff } \quad \forall s,t \in S. \distparsyn^1(s,t) \exporder \distparsyn^2(s,t).
\]
\end{definition}
As the syntactic expression computed by $\ffuncpar$ for any pair of states $(s,t)$ is merely syntactically representing the functions computed by $\ffunc$ for the same pair of states, the two concepts are closely related. For any expression $d_\expr \in \dfuncparsyn$ let $d \in \dfuncparsem$ be the associated semantic function. Then it is the case that the syntactic ordering of expressions implies the same semantic ordering of the associated semantic functions. Furthermore, iteratively updating the distances as parametric expressions by $\ffuncpar$ is semantically equivalent to computing the distances as functions by $\ffunc$.
\begin{lemma}\label{lem:synsemrelation}
For any $d_\expr^1,d_\expr^2 \in \dfuncparsyn$ and $n \in \mathbb{N}$:
\begin{enumerate}
\item $d_\expr^1 \distsynorder d_\expr^2 \implies d^1 \leq d^2$.
\item $\llbracket \ffuncpar^n(d_\expr^1)(s,t) \rrbracket (v) = \ffunc^n(d^1)(s,t)(v)$.
\end{enumerate}
\end{lemma}
We will now demonstrate an upper bound on the relevant matching transition sequences for the syntactic computations in $\ffuncpar$, given that all loops have at least one strictly positive non-parametric weight. This is similar to assuming no 0-cycles in the weighted case (\autoref{lem:finsequence}).
\begin{lemma}\label{lem:finsequence_par}
Let $\WKS$ be a WKS with state $s \in S$ such that $s \to_w s'$ and let

$\mcal{K_\param}=(S_\param,\atom,\lab_\param,\to^\param)$ be a PWKS with the following property:
\begin{itemize}
	\item There exists a $w_{\min} > 0$ such that for any valuation, the accumulated weight of every loop in $\mcal{K}_\param$ is at least $w_{\min}$ (strongly cost non-zeno).
\end{itemize}
Then for any $t \in S_\param$:
\begin{align*}
\exists N. &\forall \pi = t\rightarrow_{v_1}^\param t_1 \ldots \rightarrow_{v_n}^\param t_n, n \geq N.\\
&\exists \pi^* = t\rightarrow_{u_1}^\param t_1' \ldots \rightarrow_{u_m}^\param t_m', m \leq N.\\
&t_n = t_m' \,\land\, \left|\frac{\sum_{i=1}^m u_i}{w}-1\right| \exporder \left|\frac{\sum_{i=1}^n v_i}{w}-1\right| \land\\
& \{t_1', \ldots, t_{m-1}'\} \subseteq \{t_1,\ldots,t_{n-1}\}
\end{align*}
\end{lemma}
\begin{proof}
Let the maximum weight out of $s$ be $s_{w_{\max}}$. Any sequence of length $|S_\param|$ must have a loop which, by assumption, cannot have accumulated weight 0 w.r.t any parameter valuation. Thus, the accumulated weight w.r.t any valuation is at least $w_{\min}$. Without loss of generality we assume it to be exactly $w_{\min}$. Exercising the loop a number of times will at some point result in the accumulated weight being greater than $2 \cdot s_{w_{\max}}$ w.r.t any valuation. Let this sequence be $\pi^* = t\rightarrow_{v_1}^\param t_1 \cdots \rightarrow_{v_k}^\param  t_k$ and let $x$ denote the number of times the loop is exercised i.e $x \cdot w_{\min} \geq 2 \cdot s_{w_{\max}}$. Let $\sum_{i=1}^k v_i= \sum_{i=1}^n a_ip_i+b$. Then it is clear that $\frac{b}{s_{w_{\max}}} > 1$. Now consider the corresponding non-looping sequence $\pi_1= t\rightarrow_{u_1}^\param  t_1' \cdots \rightarrow_{u_l}^\param  t_l'$ and let $\sum_{i=1}^l u_i= \sum_{i=1}^n a_i'p_i+b'$. We would like it to be the case that 
\[
\left|\frac{\sum_{i=1}^n a'_ip_i + b'}{w}-1\right| \exporder \left|\frac{\sum_{i=1}^n a_ip_i + b}{w}-1\right|
\] 
but it might be the case that $\frac{b'}{s_{w_{\max}}} < 1$.
Consider a third sequence $\pi = t\rightarrow_{x_1}^\param t_1'' \cdots \rightarrow_{x_m}^\param  t_m''$, being $\pi^*$ modified to exercise the loop one more time and let $\sum_{i=1}^m x_i = \sum_{i=1}^n a_i''p_i+b''$. Now we know that $\frac{b''}{s_{w_{\max}}} > 1$ as $b'' > b'$ and furthermore $\left|\frac{\sum_{i=1}^n a_i'p_i + b'}{w}-1\right| \exporder \left|\frac{\sum_{i=1}^n a_i''p_i + b''}{w}-1\right|, t_k = t_m''$ and  $\{t_1',\ldots,t_k\} \subseteq \{t_1'',\ldots,t_m''\}$. We can now derive $N$. For $\pi^*$ we have the inequality $x \cdot w_{\min} \geq 2 \cdot s_{w_{\max}}$ and by \autoref{lem:finsequence} this leads to the bound $\frac{2 \cdot s_{w_{\max}}}{w_{\min}} \cdot |S_\param|$. As $\pi$ is at most $|S_\param|$ longer than $\pi^*$ we get
\[
N \geq \frac{2 \cdot s_{w_{\max}}}{w_{\min}} \cdot |S_\param| + |S_\param|
\]
\end{proof}
Note that the bound also holds for the semantic function $\ffunc$ as the syntactic ordering implies the semantic ordering (\autoref{lem:synsemrelation}).

We can now limit $\ffuncpar$ to only consider sequences of length $N$, assuming that the PWKS is strongly cost non-zeno. We apply this fact to prove that we will after a finite number of iterations of $\ffuncpar$ have discovered two syntactically equivalent expressions. As syntactic equivalence implies semantic equivalence of the associated functions, we get by \autoref{lem:synsemrelation} that $d_{\min}$ can be computed by repeated application of both $\ffunc$ and $\ffuncpar$ is a finite number of steps. 

\begin{lemma}\label{lem:synequiv}
There exists $n < m$ such that $\ffuncpar^n(\distparsyn^0) \equivexpr \ffuncpar^m(\distparsyn^0)$.
\end{lemma}
\begin{proof}
Let
\[
\ffuncpar^n(\distparsyn^0)(s,t) = \minE\left\{\maxE\left\{E_{1.1},\ldots,E_{1.k}\right\},\ldots,\maxE\left\{E_{m.1},\ldots,E_{m.n}\right\}\right\}.
\]
From the definition of $\equivexpr$ we directly get $\maxE$ and $\minE$ expressions behave like sets. Duplicates can be ignored i.e $\maxE\{E_1,E_2,E_2\} \equivexpr \maxE\{E_1,E_2\}$, $\minE\{\maxE\{E_1,E_2\},\maxE\{E_1,E_2\}\} \equivexpr \minE\{\maxE\{E_1,E_2\}\}$ and the ordering of elements does not matter; $\maxE\left\{E_1,E_2\right\} \equivexpr \maxE\left\{E_2,E_1\right\}$. By \autoref{lem:finsequence_par} we can limit the transition sequences to length $N$. This implies that only a finite number of basic elements $\left|\frac{v}{w}-1\right|$ exist when iteratively applying $\ffuncpar$. As one can only construct a finite number of unique sets from a finite set of elements, the number of syntactically unique expressions (w.r.t $\equivexpr$) is finite. Therefore, there must exist a $m > n$ such that $\ffuncpar^n(\distparsyn^0) \equivexpr \ffuncpar^m(\distparsyn^0)$.
\end{proof}
We can now demonstrate computability of the distance.
\begin{theorem}[Computability]
There exists a natural number $n$ such that for all states $s \in S, t \in S_\param$ and all valuations $v \in \eval$
\[
\llbracket \ffunc^n(\distparsyn^0)(s,t) \rrbracket (v) = d_{\min}(s,t)(v).
\]
\end{theorem}
\begin{proof}
By \autoref{lem:synequiv}, there exists $n < m$ such that $\ffuncpar^n(\distparsyn^0) \equivexpr \ffuncpar^m(\distparsyn^0)$. By \autoref{lem:synsemrelation} we thus get semantic equivalence $\ffunc^n(d^0) \equiv \ffunc^m(d^0)$ and as $\ffunc$ is monotonic on $(\dfuncparsem, \leq)$ we have for all $i$ s.t $n \leq i \leq m$ that $\ffunc^i(d^0) \equiv \ffunc^m(d^0)$. Thus, $\ffunc^n(d^0)$ is a fixed point found after a finite number of steps and is captured syntactically by $\ffuncpar^n(\distparsyn^0)$. The check for equivalence ($\equivexpr$) can therefore be used to capture a semantic fixed point syntactically. The fixed point must also be the least fixed point. To see this, suppose towards a contradiction that it is not the least fixed point. Then there exists a $k < n$ such that $\ffunc^k(d^0) = d_{\min}$ but by the fixed point property of $d_{\min}$ and the monotonicity of $\ffunc$ we immediately get $\ffunc^k(d^0) \equiv \ffunc^n(d^0)$ which contradicts our assumption that $\ffunc^n(d^0)$ is not the least fixed point of $\ffunc$.
\end{proof}
By computing the syntactic fixed point we thus get a syntactic expression
$\ffuncpar^n(\distparsyn^0)(s,t) = E_{s,t}$ for each pair of states $s,t$ such that the solution set to $E_{s,t} \leq \veps$ characterizes the set of ``good'' parameter valuations.
\begin{example}
Consider the WKS and PWKS from \autoref{fig:PWKSex}. To compute $E_{s,t}$, let $\distparsyn^i(s,t) = \ffuncpar^i(\distparsyn^0)(s,t)$. We now show how the distance from $s$ to $t$ is updated after each iteration.
\begin{alignat*}{2}
&\distparsyn^1(s,t) &&= \maxE\left\{\begin{array}{ll}
				                     \maxE\left\{\left|\frac{1}{1}-1\right|,\distparsyn^0(s_1,t_2)\right\},\\[0.15cm]
														 \maxE\left\{\left|\frac{3}{2}-1\right|,\distparsyn^0(s_1,t_2),d^0_\expr(s,t_2)\right\}
														 \end{array}\right\}\\
&\distparsyn^2(s,t) &&= \maxE\left\{\begin{array}{ll}
				                     \maxE\left\{\left|\frac{1}{1}-1\right|,0\right\},\\[0.15cm]
														 \maxE\left\{\left|\frac{3}{2}-1\right|,0,\frac{1}{2}\right\}
														 \end{array}\right\}\\
&\distparsyn^3(s,t) &&= \maxE\left\{\begin{array}{l}
															\frac{1}{2},\left|\frac{p}{1}-1\right|,\\
															\minE\left\{\left|\frac{p}{5}-1\right|,\left|\frac{p+2}{5}-1\right|,\left|\frac{p+4}{5}-1\right|\right\},\\[0.15cm]
															\minE\left\{\left|\frac{p}{3}-1\right|,\left|\frac{p+2}{3}-1\right|\right\}
															\end{array}\right\}\\
&\distparsyn^4(s,t) && = \distparsyn^3(s,t)
\end{alignat*}
We immediately see that any solution to $E_{s,t} \leq \veps$ is bounded from below by $\frac{1}{2}$. This implies that there exists no valuation $v \in \eval$ such that $s \simeps t[v]$ for $\veps < \frac{1}{2}$. If we consider the valuation $v_{\min}(p) = 1$ we get that $\llbracket E_{s,t} \rrbracket (v_{\min}) = \frac{1}{2}$  i.e $v_{\min}$ is the valuation that induces the minimal distance $d(s,t[v_{\min}]) = \frac{1}{2}$.
\end{example}
\section{Conclusion and Future Work}\label{sec:conc_future}
We have characterized the distance from \cite{WCTL_logic} between weighted Kripke structures (WKS) as a least fixed point. The distance between any pair of states can thus be computed by first assuming the distance between any pair to be 0 and then applying a step-wise refinement of the distance. The computability of the distance is guaranteed as a finite number of the (potentially) infinite transition sequences of the system is sufficient. This we proved by demonstrating an upper bound on the relevant sequences. We furthermore lifted the distance to parametric WKS (PWKS), where transition weights can be parametric. The parameters can be used to abstract multiple configurations of the same system as one parametric system. In this case the distance is from a WKS to a PWKS and is concretely a parametric expression that one can evaluate to get an exact distance from the WKS to a specific WKS instance of the PWKS.  The question is then which configuration (parameter valuation) is ``best'' i.e minimizes the induced distance. For computability we again demonstrate an upper bound on the length of relevant distances. To do this we assume all cycles to be cost non-zeno i.e any loop must include a transition with a positive rational weight.

For future work, the actual complexity of computing the distance is open. From \cite{WCTL_logic} we know that checking whether the distance is 0 is NP-complete but the general complexity of checking whether the distance is less than some $\veps \in \Rpos$ is open. One could also investigate whether the distance has a polynomial approximation scheme.
\bibliographystyle{eptcs}
\bibliography{refs}
\end{document}